%% file: final.tex
\documentclass[letterpaper, 10pt, conference]{ieeeconf}
\IEEEoverridecommandlockouts
\usepackage{geometry}
\usepackage{graphicx}
\usepackage{epsfig}
\usepackage{amsmath}

\usepackage{amssymb} 
\usepackage[noadjust]{cite}
\usepackage{bm}
\usepackage{subfigure}
\usepackage{acronym}
\usepackage{paralist}
\usepackage{float}
\usepackage{epstopdf}
\usepackage{algorithm}
\usepackage[noend]{algpseudocode}
\usepackage{mathtools}
\usepackage{multicol} 
\usepackage{accents}
\newcommand{\subparagraph}{}
\usepackage{titlesec}

\usepackage{amsthm}

\makeatletter
\def\BState{\State\hskip-\ALG@thistlm}
\makeatother

\include{JUcmds}

\theoremstyle{plain}
\newtheorem{problem}{Problem}

\newcommand{\rmax}{r_{\max}}
\newcommand{\bsig}{\bm \sigma}

\makeatletter
\renewcommand{\ALG@beginalgorithmic}{\footnotesize}
\makeatother

\def\figsc{1}

\titlespacing\section{0pt}{11pt plus 4pt minus 2pt}{2pt plus 2pt minus 2pt}

\begin{document}

\title{\LARGE \bf Determining r-Robustness of Digraphs Using Mixed Integer Linear Programming}


\author{James Usevitch and Dimitra Panagou
\thanks{The authors are with the Department of Aerospace Engineering, University of Michigan, Ann Arbor; \texttt{usevitch@umich.edu}, \texttt{dpanagou@umich.edu}.}
\thanks{The authors would like to acknowledge the support of the Automotive Research Center (ARC) in accordance with Cooperative Agreement W56HZV-14-2-0001 U.S. Army TARDEC in Warren, MI, and the Award No W911NF-17-1-0526.}
}

\maketitle
\thispagestyle{empty}
\pagestyle{empty}

\acrodef{wrt}[w.r.t.]{with respect to}
\acrodef{apf}[APF]{Artificial Potential Fields}
\begin{abstract}
Convergence guarantees of many resilient consensus algorithms are based on the graph theoretic properties of $r$- and $(r,s)$-robustness. These algorithms guarantee consensus of normally behaving agents in the presence of a bounded number of arbitrarily misbehaving agents if the values of the integers $r$ and $s$ are sufficiently high. However, determining the largest integer $r$ for which an arbitrary digraph is $r$-robust is highly nontrivial. This paper introduces a novel method for calculating this value using mixed integer linear programming. The method only requires knowledge of the graph Laplacian matrix, and can be formulated with affine objective and constraints, except for the integer constraint. Integer programming methods such as branch-and-bound can allow both lower and upper bounds on $r$ to be iteratively tightened. Simulations suggest the proposed method demonstrates greater efficiency than prior algorithms.

\end{abstract}

\IEEEpeerreviewmaketitle

\section{Introduction}
\label{intro}

Consensus on shared information is fundamental to the operation of multi-agent systems. In context of mobile agents, it enables formation control, agent rendezvous, sensor fusion, and many more objectives. Although a vast literature of algorithms for consensus exist, many are unable to tolerate the presence of adversarial attacks or faults. 
Recent years have seen an increase of attention on \emph{resilient} algorithms that are able to operate despite such misbehavior. 
Many of these algorithms have been inspired by work such as \cite{Lamport1982}, which is one of the seminal papers on consensus in the presence of adversaries; \cite{leblanc2013resilient,Zhang2012robustness,LeBlanc_2013_Res_Continuous} which outline discrete- and continuous-time algorithms along with necessary and sufficient conditions for scalar consensus in the presence of Byzantine adversaries; and \cite{Vaidya2012iterative,Vaidya2013byzantine,Tseng2013iterative,Tseng2014asynchronous}, which outline algorithms for multi-agent vector consensus of asynchronous systems in the presence of Byzantine adversaries. Some of the most recent results that draw upon these works include resilient state estimation \cite{Mitra2018secure}, resilient rendezvous of mobile agents \cite{Park2017fault,park2016efficient}, resilient output synchronization \cite{leblanc2017resilient}, resilient simultaneous arrival of interceptors \cite{li2018robust}, resilient distributed optimization \cite{sundaram2018distributed,su2016fault}, reliable broadcast \cite{tseng2015broadcast,Zhang2012robustness}, and resilient multi-hop communication \cite{su2017reaching}.


Many of these results are based upon the graph theoretical properties known as $r$-robustness and $(r,s)$-robustness \cite{leblanc2013resilient,Zhang2012robustness}. These notions were defined after it was shown that traditional graph theoretic metrics (e.g. connectivity) were insufficient to analyze the convergence properties of certain resilient algorithms based on purely local information \cite{Zhang2012robustness}. The properties of $r$- and $(r,s)$-robustness 
have been used in sufficient conditions for several resilient consensus algorithms including the ARC-P \cite{LeBlanc_2013_Res_Continuous}, W-MSR \cite{leblanc2013resilient}, SW-MSR \cite{saldana2017resilient}, and DP-MSR \cite{dibaji2017resilient} algorithms. Given an upper bound on the global or local number of adversaries in the network, these resilient algorithms guarantee convergence of normally behaving agents' states to a value within the convex hull of initial states if the integers $r$ and $s$ are sufficiently large.

A key challenge in implementing these resilient algorithms is that determining the $r$- and $(r,s)$-robustness of arbitrary digraphs is an NP-hard problem in general \cite{leblanc2013algorithms,zhang2015notion}. 
The first algorithmic analysis of determining the values of $r$ and $s$ for arbitrary digraphs was given in \cite{leblanc2013algorithms}. The algorithms in this work employ an exhaustive search to determine the maximum values of $r$ and $s$ for a given digraph, 
and have exponential complexity w.r.t. the number of nodes in the network. 
Subsequent work has focused on methods to circumvent this difficulty, including
graph construction methods which increase the graph size while preserving initial values of $r$ and $s$ \cite{leblanc2013resilient,Guerrero2016formations},
demonstrating the behavior of $r$ as a function of particular graph properties 
\cite{zhang2015notion,shahrivar2017spectral,zhao2017connectivity},
lower bounding $r$ with the isoperimetric constant and algebraic connectivity of undirected graphs \cite{shahrivar2015robustness},
and even using
machine learning to correlate characteristics of
certain graphs to the values of $r$ and $s$ \cite{wang2018using}. Finding more efficient ways of determining the \emph{exact} robustness of digraphs however is still an open problem.

In this paper, we introduce a novel method for determining the maximum value of $r$ for which an arbitrary digraph is $r$-robust by solving a mixed integer linear programming (MILP) problem. The problem only requires knowledge of the graph Laplacian matrix and can be formulated with affine objective and constraints, with the exception of the integer constraint.
To the best of our knowledge, this is the first time the problem has been formulated in this way. This contribution provides several advantages. First, these results open the door for the extensive literature on integer programming to be applied to the $r$-robustness determination problem. In particular, applying branch-and-bound algorithms to the problem can allow for lower and upper bounds on a digraph's $r$-robustness to be iteratively tightened. Prior algorithms are only able to tighten the upper bound on the maximum robustness for a given digraph.
Second, this formulation enables commercially available solvers such as Gurobi or MATLAB's \emph{intlinprog} to be used to find the maximum robustness of any digraph. Finally, experimental results using this new formulation suggest a reduction in computation time as compared to the centralized algorithm proposed in \cite{leblanc2013algorithms}.

This paper is organized as follows: notation and relevant definitions are introduced in Section \ref{sec:notation}. The problem formulation is given in Section \ref{sec:problemformulation}. Our main result of formulating the $r$-robustness determination problem as a mixed integer linear programming problem is given in Section \ref{sec:rrobustdeterm}. Simulations are presented in Section \ref{sec:simulations}, and we present conclusions and directions for future work in Section \ref{sec:conclusion}.

\section{Notation}
\label{sec:notation}

The real numbers and integers are denoted $\R$ and $\Z$, respectively. The  nonnegative real numbers and integers are denoted $\R_+$ and $\Z_+$, respectively. $\R^n$ denotes an $n$-dimensional vector space over the field $\R$, $\Z^n$ represents an $n$ dimensional vector with nonnegative integer vectors, and $\{0,1\}^n$ represents a binary vector of length $n$. Scalars are denoted in normal text (e.g. $x \in \R$) while vectors are denoted in bold (e.g. $\bm x \in \R^n$).
The notation $x_i$ denotes the $i$th entry of vector $\bm x$. 
The inequality symbol $\gleq$ denotes a componentwise inequality between vectors; i.e. for $\bm x,\bm y \in \R^n$, $\bm x \gleq \bm y \implies x_i \leq  y_i\ \forall i \in \{1,\ldots,n\}$. An $n$-dimensional vector of ones is denoted $\bm 1_n$, and an $n$-dimensional vector of zeros is denoted $\bm 0_n$. In both cases the subscript $n$ will be omitted when the size of the vector is clear from the context. The union, intersection, and set complement operations are denoted by $\cup,\ \cap$, and $\setminus$, respectively. 
The cardinality of a set is denoted as $|S|$, and the empty set is denoted $\{ \emptyset \}$. The infinity norm on $\R^n$ is denoted $\nrm{\cdot}_\infty$. The notation $C(n,k) = n!/(k!(n-k)!)$ denotes the binomial coefficient with $n,k \in \Z_+$. Given a set $S$, the power set of $S$ is denoted $\Pc(S) = \{A : A \subseteq S \}$.

A directed graph (digraph) is denoted as $\D = (\V,\E)$, where $\V = \{1,\ldots,n\}$ is the set of indexed nodes and $\E$ is the edge set. 
A directed edge is denoted $(i,j)$, with $i,j \in \V$. The set of in-neighbors for an agent $j$ is denoted $\N_j = \{i \in \V : (i,j) \in \E \}$. The minimum in-degree of a digraph $\D$ is denoted $\delta^{in}(\D) = \min_{j \in \V} |\N_j|$.
In this paper we consider \emph{simple} digraphs of $n$ nodes, meaning digraphs without self loops $\big((i,i) \notin \E\ \forall i \in \V \big)$ and without redundant edges (i.e. at most one directed edge $(i,j) \in \E$ exists from $i$ to $j$).
Occasionally, $\G = (\V,\E)$ will be used to denote an undirected graph where $(i,j) \in \E \iff (j,i) \in \E$ $\forall i,j \in \V$.
The graph Laplacian $L$ for a digraph (or undirected graph) is defined as follows, with $L_{j,i}$ denoting the entry in the $j$th row and $i$th column:

\begin{equation}
\label{eq:Laplacian}
    L_{j,i} = \begin{cases}
        |\N_j| & \text{if } j = i \\
        -1 & \text{if } i \in \N_j \\
        0 & \text{if } i \notin \N_j
    \end{cases}
\end{equation}

\section{Problem Formulation}
\label{sec:problemformulation}

We begin with the definitions of $r$-reachability and $r$-robustness:

\begin{define}[\cite{leblanc2013resilient}]
Let $r \in \Z_+$ and $\D=(\V,\E)$ be a digraph. A nonempty subset $S \subset \V$ is $r$-reachable if $\exists i \in S$ such that $|\N_i \backslash S| \geq r$.
\end{define}


\begin{define}[\cite{leblanc2013resilient}]
\label{def:rrobust}
Let $r \in \Z_+$. A nonempty, nontrivial digraph $\D = (\V,\E)$ on $n$ nodes $(n \geq 2)$ is $r$-robust if for every pair of nonempty, disjoint subsets of $\V$, at least one of the subsets is $r$-reachable. By convention, the empty  graph $(n = 0)$ is 0-robust and the trivial graph $(n=1)$ is 1-robust.
\end{define}
If a set $S$ is $r$-reachable, it is $r'$-reachable for any $0 \leq r' \leq r$.
Similarly, if a graph is $r$-robust it is also $r'$-robust for any $0 \leq r' \leq r$. 

This paper addresses the following problem:
\begin{problem}
\label{prob:robust}
    Given an arbitrary digraph $\D$, determine the maximum integer $r$ for which $\D$ is $r$-robust.
\end{problem}

\begin{remark}
\label{rmk:rho}
    We denote the maximum integer $r$ for which a given digraph $\D$ is $r$-robust as $r_{\max}(\D) \in \Z_+$.
\end{remark}
It should be clear from Definition \ref{def:rrobust} that determining $\rmax(\D)$ involves checking the reachability of pairs of nonempty, disjoint subsets in a graph. Let the set $\mathcal{T} \subset \Pc(\V) \times \Pc(\V)$ be defined as 
\begin{align}
        \mathcal{T} =  \big\{ (S_1,S_2) \in \Pc(\V) \times \Pc(\V) : |S_1| > 0,\ |S_2| > 0,\ & \nonumber \\
        |S_1 \cap S_2| = 0 \big\}& \label{eq:Tdef}
\end{align}
The set $\mathcal{T}$ therefore contains all possible pairs of nonempty, disjoint subsets of $\V$. It was shown in \cite{leblanc2013algorithms} that $|\mathcal{T}| = \sum_{p=2}^n \pmxs{n \\ p} (2^{p} -2)$.\footnote{Since $(S_1,S_2) \in \mathcal{T} \implies (S_2,S_1) \in \mathcal{T}$, the total number of \emph{unique} nonempty, disjoint subsets is $(1/2)|\mathcal{T}|$, denoted as $R(n)$ in \cite{leblanc2013algorithms}.}

\subsection{Alternate Formulation of Maximum $r$-Robustness}

In our first result, we derive an equivalent way of expressing the maximum robustness $r_{\max}(\D)$ of a digraph $\D$.
Given an arbitrary digraph $\D=(\V,\E)$ and a subset $S \subset V$, we define the reachability function $\Rc : \Pc(\V) \rarr \Z_+$ as follows:
\begin{align}
\label{eq:reachmax}
    \Rc(S) &= \begin{cases}
    \max_{i \in S } |\N_i \backslash S | & \text{if } S \neq \{\emptyset\} \\
    0 & \text{if } S = \{\emptyset\}
    \end{cases}
\end{align}
In other words, $\Rc(S)$ returns the maximum integer $r$ for which the set $S$ is $r$-reachable. 
The following Lemma presents an explicit formulation which yields $r_{\max}(\D)$:

\begin{lemma}
\label{lem:rrobalt}
    Let $\D = (\V,\E)$ be an arbitrary nonempty, nontrivial, simple digraph with $|\V| = n$. The following holds:
    \begin{equation}
    \label{eq:rrobalt}
        \begin{aligned}
            r_{\max}(\D) =& \underset{S_1,S_2 \in \Pc(\V)}{\text{minimize}}
            & & \max\pth{\Rc(S_1) , \Rc(S_2)} \\
            & \text{subject to}
            & & |S_1| > 0,\ |S_2| > 0,\ |S_1 \cap S_2| = 0 \\
        \end{aligned}
    \end{equation}
\end{lemma}

\begin{proof}
    For brevity, define the function
    \begin{equation*}
        g(S_1,S_2) = \max\pth{\Rc(S_1) , \Rc(S_2)},\ g: \Pc(\V) \times \Pc(\V) \rarr \Z_+
    \end{equation*}
    Note that $g(S_1,S_2) = m \implies \Rc(S_1) = m$ or $\Rc(S_2) = m$. Let $(S_1^*, S_2^*)$ be a minimizer of the right hand side (RHS) of \eqref{eq:rrobalt}. Then $g(S_1^*,S_2^*) \leq g(S_1,S_2)\ \forall (S_1,S_2) \in \Tc$. Therefore $\forall (S_1,S_2) \in \Tc,$ either $\Rc(S_1) \geq g(S_1^*,S_2^*)$ or $\Rc(S_2) \geq g(S_1^*,S_2^*)$. Therefore the graph is $g(S_1^*,S_2^*)$-robust, implying $r_{\max}(\D) \geq g(S_1^*,S_2^*)$.
    
    To show $r_{\max}(\D) = g(S_1^*,S_2^*)$, we prove by contradiction. By definition, $\D$ is $r_{\max}(\D)$-robust (Remark \ref{rmk:rho}). Suppose $r_{\max}(\D) > g(S_1^*,S_2^*)$. This implies $\Rc(S_1^*) < r_{\max}(\D)$ and $\Rc(S_2^*) < r_{\max}(\D)$. Since $\exists (S_1^*,S_2^*) \in \Tc$ such that $\Rc(S_1^*) < r_{\max}(\D)$ and $\Rc(S_2^*) < r_{\max}(\D)$, this implies $\D$ is not $r_{\max}(\D)$-robust. This contradicts the fact that $\D$ is $r_{\max}(\D)$-robust by definition. Therefore, $r_{\max}(\D) = g(S_1^*,S_2^*)$.
\end{proof}

\begin{remark}
    \label{rmk:implicit}
    Using the definition of $\mathcal{T}$ in \eqref{eq:Tdef}, the constraints on the RHS of \eqref{eq:rrobalt} can be made implicit \cite[section 4.1.3]{boyd2004convex} as follows:
    \begin{equation}
        \begin{aligned}
            r_{\max}(\D) =& \underset{(S_1,S_2) \in \mathcal{T}}{\text{minimize}}
            & & \max\pth{\Rc(S_1) , \Rc(S_2)} 
        \end{aligned}
    \end{equation}
\end{remark}


\section{$r$-Robustness Determination as an MILP}
\label{sec:rrobustdeterm}


The next step in the analysis is to demonstrate how the expression $\max\pth{\Rc(S_1) , \Rc(S_2)}$
can be calculated as a function of the graph Laplacian matrix.
Recall that $n = |\V|$, and define the indicator vector $\bm \sigma(\cdot) : \Pc(\V) \rarr \{0,1\}^{n}$ as follows: for any $S \in \Pc(\V)$,
\begin{align}
\label{eq:sigmadef}
    \bm \sigma_j(S) &= \begin{cases}
        1 & \text{if } j \in S \\
        0 & \text{if } j \notin S
    \end{cases},\ j = \{1,\ldots,n \}
\end{align}

In other words the $j$th entry of $\bm \sigma(S)$ is 1 if the node with index $j$ is a member of the set $S \in \Pc(\V)$, and zero otherwise. It is straightforward to verify that $\bm \sigma(\cdot)$ is a bijection. Therefore given $\bm x \in \{0,1\}^n$, the set $\bm \sigma^{-1}(\bm x) \in \Pc(\V)$ is defined by $\bm x_j = 1 \implies j \in \bm \sigma^{-1}(\bm x)$ and $\bm x_j = 0 \implies j \notin \bm \sigma^{-1}(\bm x)$.

\begin{lemma}
\label{lem:Ljsig}
    Let $\D= (\V,\E)$ be an arbitrary nonempty, nontrivial, simple digraph, let $L$ be the Laplacian matrix of $\D$, and let $S \in \Pc(\V)$. Then the following holds for all $j \in \{1,\ldots,n\}$:
    \begin{align}
    	\label{eq:Avicii}
        L_j \bm \sigma(S) &= \begin{cases}
            |\N_j \backslash S|, & \text{if } j \in S, \\
            -|\N_j \cap S|, & \text{if } j \notin S,
        \end{cases}
    \end{align}
    where $L_j$ is the $j$th row of $L$. Furthermore,
    \begin{align}
    	\label{eq:Coldplay}
    	\Rc(S) = \max_j L_j \bm \sigma(S),\ j \in \{1,\ldots,n\}.
    \end{align}
\end{lemma}
\begin{proof}
The term $\bm \sigma(S)$ is shortened to $\bm \sigma$ for brevity. Recall that the entry in the $j$th row and $i$th column of $L$ is denoted $L_{j,i}$. The definition of $L$ from \eqref{eq:Laplacian} implies
    \begin{align}
    	L_j \bm \sigma &= (L_{j,j}) \sigma_j + \sum_{q \in \{1,\ldots,n\} \backslash j } (L_{j,q}) \sigma_q  \nonumber \\
         &= |\N_j| \sigma_j   - \sum_{q \in \N_j \cap S } \sigma_q  - \sum_{q \in \N_j \backslash S } \sigma_q. \label{eq:superman}
    \end{align}
Since by \eqref{eq:sigmadef}, $q \in S$ implies $\sigma_q = 1$, the term $\sum_{q \in \N_j \cap S } \sigma_q = |\N_j \cap S|$. In addition, since $q \notin S$ implies $\sigma_q = 0$, the term $\sum_{q \in \N_j \backslash S } \sigma_q = 0$. By this, equation \eqref{eq:superman} simplifies to $L_j \bm  \sigma  = |\N_j| \sigma_j   - |\N_j \cap S |$.

The value of the term $|\N_j| \sigma_j$ depends on whether $j \in S$ or $j \notin S$. If $j \in S$, then $\sigma_j = 1$, implying $ L_j \bm  \sigma  = |\N_j| - |\N_j \cap S | = \pth{|\N_j \cap S | + |\N_j \backslash S |} -  |\N_j \cap S | = |\N_j \backslash S |$.
If $j \notin S $, then $\sigma_j = 0$ implying $L_j \bm  \sigma  = -|\N_j \cap S |$. This proves the result for equation \eqref{eq:Avicii}.

To prove \eqref{eq:Coldplay}, we first consider nonempty sets $S \in \Pc(\V) \backslash \{\emptyset\}$. By the results above and \eqref{eq:reachmax}, the maximum reachability of any $S \in \Pc(\V) \backslash \{ \emptyset \}$ is found by
\begin{equation}
    \Rc(S) = \max_{j \in  S} |\N_j \backslash S| = \max_{j \in S} (L_j \bm \sigma(S)). \label{eq:juicyfruit}
\end{equation}
By its definition, $\Rc(S) \geq 0$. Observe that
if $j \in S$ then $L_j \bm \sigma(S) = |\N_j \backslash S| \geq 0$, implying $\max_{j \in S} L_j \bm \sigma(S) \geq 0$.
Conversely,
if an agent $j$ is \emph{not} in the set $S$, then the function $L_j \bm \sigma(S)$ takes the nonpositive value $-|\N_j \cap S|$. This implies $\max_{j \notin S} L_j \bm \sigma(S) \leq 0$.
By these arguments, we therefore have $\max_{j \notin S} L_j \bm \sigma(S) \leq 0 \leq \max_{j \in S} L_j \bm \sigma(S)$, which implies
\begin{align}
	\max_{j \in \{1,\ldots,n \}} L_j \bm \sigma(S) &= \max\pth{(\max_{j \in S} L_j \bm \sigma(S)), (\max_{j \notin S} L_j \bm \sigma (S))} \nonumber \\
	&= \max_{j \in S} L_j \bm \sigma(S). \label{eq:gingerale}
\end{align}
Therefore by equations \eqref{eq:gingerale} and \eqref{eq:juicyfruit}, the maximum reachability of $S$ is found by the expression
\begin{equation}
    \label{eq:reachS}
    \Rc(S) = \max_j (L_j \bm \sigma(S)),\ j \in \{1,\ldots,n\}.
\end{equation}
Lastly, if $S = \emptyset$, then by \eqref{eq:reachmax} we have $\Rc(S) = 0$. In addition, $\bm \sigma(S) = \bm 0$, implying that $\max_j L_j \bm \sigma(S) = 0 = \Rc(S),\ j \in \{1,\ldots,n\}$.    
\end{proof}

Lemma \ref{lem:Ljsig} can be used to rewrite the objective function of \eqref{eq:rrobalt} in terms of the Laplacian matrix of $\D$:

\begin{lemma}
\label{lem:maxequal}
Let  $\D = (\V,\E)$ be an arbitrary nonempty, nontrivial, simple digraph. Let $L$ be the Laplacian matrix of $\D$, and let $L_j$ be the $j$th row of $L$. Let $\Tc$ be defined as in \eqref{eq:Tdef}. Then for all $(S_1,S_2) \in \Tc$ the following holds:
\begin{align}
    &\max \pth{\Rc(S_1),\Rc(S_2)} = \nonumber \\
     &\max \pth{ \max_i \pth{L_i \bm \sigma(S_1)}, \max_j \pth{L_j \bm \sigma(S_2)} } \\
     &\hspace{.5em} i,j \in \{1,\ldots,n\}. \nonumber
\end{align}
\end{lemma}


\begin{proof}
By Lemma \ref{lem:Ljsig}, $\Rc(S_1) = \max_i L_i \bm \sigma(S_1)$ for $i \in \{1,\ldots,n\}$, and $\Rc(S_2) = \max_j L_j \bm \sigma(S_2)$ for $j \in \{1,\ldots,n\}$. The result follows.    
\end{proof}
\vspace{-.5em}
From Lemma \ref{lem:rrobalt}, Lemma \ref{lem:maxequal}, and Remark \ref{rmk:implicit},
we can immediately conclude that $\rmax(\D)$ satisfies
\begin{align}
    \label{eq:rrobalt5000}
            &\rmax(\D) = \nonumber \\
            &\underset{(S_1,S_2) \in \Tc}{\min}
             \max \big(\max_i \pth{L_i \bm \sigma(S_1)}, \max_j \pth{L_j \bm \sigma(S_2)} \big).
\end{align}
Note that the terms $\bm \sigma(S_1)$ and $\bm \sigma(S_2)$ are each $n$-dimensional binary vectors. Letting $\bm b^1 = \bm \sigma(S_1)$ and $\bm b^2 = \bm \sigma(S_2)$, the objective function of \eqref{eq:rrobalt5000} can be written as $\max \big(\max_i \pth{L_i \bm b^1}, \max_j \pth{L_j \bm b^2} \big)$. Every pair $(S_1,S_2) \in \Tc$ can be mapped into a pair of binary vectors $(\bm b^1, \bm b^2)$ by the function $\Sigma : \Tc \rarr \{0,1\}^n \times \{0,1\}^n$, where $\Sigma(S_1,S_2) = (\bm \sigma(S_1), \bm \sigma(S_2)) = (\bm b^1, \bm b^2)$.
By determining the image of $\Tc$ under $\Sigma(\cdot,\cdot)$, the optimal value of \eqref{eq:rrobalt5000} can be found by minimizing over pairs of binary vectors $(\bm b_1, \bm b_2) \in \Sigma(\Tc)$ directly. Using binary vector variables instead of set variables $(S_1,S_2)$ will allow \eqref{eq:rrobalt5000} to be written directly in a MILP form.
%
%
Towards this end, the following Lemma defines the set $\Sigma(\Tc)$:

\begin{lemma}
\label{lem:bijec}
    Let $\D = (\V,\E)$ be an arbitrary nonempty, nontrivial, simple digraph, and let $\Tc$ be defined as in \eqref{eq:Tdef}.
    Define the function $\Sigma: \mathcal{T} \rarr \{0,1\}^n \times \{0,1\}^n$ as
    \begin{equation}
    \label{eq:Sigma}
        \Sigma(S_1,S_2) = (\bm \sigma(S_1),\bm \sigma(S_2)),\ (S_1,S_2) \in \Tc.
    \end{equation}
    Define the set $\Bc \subset \{0,1\}^n \times \{0,1\}^n$ as
    \begin{align}
        \mathcal{B} = \bigg\{&  (\bm b^1, \bm b^2) \in \{0,1 \}^n \times \{0,1 \}^n : 1 \leq \bm 1^T \bm b^1 \leq (n-1), \nonumber \\
         &1 \leq \bm 1^T \bm b^2 \leq (n-1),\
         \bm b^1 + \bm b^2 \gleq \bm 1 \bigg\}. \label{eq:Bset}
    \end{align}
    Then both of the following statements hold:
    \begin{enumerate}
    \item The image of $\Tc$ under $\Sigma$ is equal to $\Bc$, i.e. $\Sigma(\Tc) = \Bc$
    \item The mapping $\Sigma : \Tc \rarr \Bc$ is a bijection.
    \end{enumerate}
\end{lemma}


\begin{proof}
We prove \emph{1)} by showing first that $\Sigma(\Tc) \subseteq \Bc$, and then $\Bc \subseteq \Sigma(\Tc)$. Any $(S_1,S_2) \in \Tc$ satisfies $|S_1| > 0$, $|S_2| > 0$, $|S_1 \cap S_2| = 0$ as per \eqref{eq:Tdef}. Observe that
\begin{align*}
	|S_1| > 0 &\implies \bm 1^T \bsig(S_1) \geq 1, \\
	|S_2| > 0 &\implies \bm 1^T \bsig(S_2) \geq 1.
\end{align*}
Because $S_1,S_2 \subset \V$ and $|S_1 \cap S_2| = 0$, then $|S_1| < n$. Otherwise if $|S_1| = n$ then either $|S_2| = 0$ or $|S_1 \cap S_2| \neq 0$, which both contradict the definition of $\Tc$. Therefore $|S_1| < n$, and by similar arguments $|S_2| < n$. Observe that
\begin{align*}
	|S_1| < n &\implies \bm 1^T \bsig(S_1) \leq n- 1, \\
	|S_2| < n &\implies \bm 1^T \bsig(S_2) \leq n-1.
\end{align*}
Finally, $|S_1 \cap S_2| = 0$ implies that $j \in S_1 \implies j \notin S_2$ and $j \in S_2 \implies j \notin S_1$ $\forall j \in \{1,\ldots,n\}$. Therefore $\sigma_j(S_1) = 1 \implies \sigma_j(S_2) = 0$ and $\sigma_j(S_2) = 1 \implies \sigma_j(S_1) = 0$.
This implies that
\begin{align*}
	|S_1 \cap S_2| = 0 \implies \bm \sigma(S_1) + \bm \sigma(S_2) \gleq \bm 1.
\end{align*}
Therefore for all $(S_1,S_2) \in \Tc$,
$(\bm \sigma(S_1), \bm \sigma(S_2)) = $ $\Sigma(S_1,S_2)$ satisfies the constraints of the set on the RHS of \eqref{eq:Bset}. This implies that $\Sigma(\Tc) \subseteq \Bc$.

Next, we show $\Bc \subseteq \Sigma(\Tc)$ by showing that for all $(\bm b^1, \bm b^2) \in \Bc$,
there exists an $(S_1,S_2) \in \Tc$ such that $(\bm b^1, \bm b^2) = \Sigma(S_1,S_2)$. Choose any $(\bm b^1, \bm b^2) \in \Bc$ and define sets $(S_1,S_2)$ as follows
for $j \in \{1,\ldots,n\}$:
\begin{align}
	b^1_j = 1 &\implies j \in S_1, &b^2_j = 1 &\implies j \in S_2, \nonumber \\
	b^1_j = 0 &\implies j \notin S_1, &b^2_j = 0 &\implies j \notin S_2. \label{eq:Medusa}
\end{align}
For the considered sets $(S_1,S_2)$, $1 \leq \bm 1^T \bm b^1$ implies $ |S_1| > 0$ and $1 \leq \bm 1^T \bm b^2$ implies $ |S_2| > 0$. In addition since $\bm b^1 + \bm b^2 \gleq \bm 1$, we have $b^1_j = 1 \implies b^2_j = 0$ and $b^2_j = 1 \implies b^1_j = 0$.
By our choice of $S_1$ and $S_2$, we have $b_j^1 = 1 \implies j \in S_1$, and from previous arguments $b_j^1 = 1 \implies b_j^2 = 0 \implies j \notin S_2$. Similar reasoning can be used to show that $b_j^2 = 1 \implies j \notin S_1$. 
These arguments imply that $|S_1 \cap S_2| = 0$. 
Consequently, $(S_1,S_2)$ satisfies all the constraints of $\Tc$ and is therefore an element of $\Tc$. Clearly, by \eqref{eq:Medusa} we have $\Sigma(S_1,S_2) = (\bm b^1,\bm b^2)$, which shows that there exists an $(S_1,S_2) \in \Tc$ such that $(\bm b^1, \bm b^2) = \Sigma(S_1,S_2)$. Since this holds for all $(\bm b^1, \bm b^2) \in \Bc$, this implies $\Bc \subseteq \Sigma(\Tc)$. Therefore $\Sigma(\Tc) = \Bc$.

We now prove \emph{2)}. Since $\Sigma(\Tc) = \Bc$, the function $\Sigma : \Tc \rarr \Bc$ is surjective.
To show that it is injective, consider any $\Sigma(S_1,S_2) \in \Bc$ and $\Sigma(\bar{S}_1, \bar{S}_2) \in \Bc$ such that $\Sigma(S_1,S_2) = \Sigma(\bar{S}_1, \bar{S}_2)$. This implies $(\bsig(S_1),\bsig(S_2)) = (\bsig(\bar{S}_1), \bsig(\bar{S}_2))$. Note that $(\bsig(S_1),\bsig(S_2)) = (\bsig(\bar{S}_1), \bsig(\bar{S}_2))$ if and only if $\bsig(S_1) = \bsig(\bar{S}_1)$ and $\bsig(S_2) = \bsig(\bar{S}_2)$. Since the indicator function $\bm \sigma : \Pc(\V) \rarr \{0,1\}^n$ is itself injective, this implies $S_1 = \bar{S}_1$ and $S_2 = \bar{S}_2$, which implies $(S_1,S_2) = (\bar{S}_1,\bar{S}_2)$. Therefore $\Sigma : \Tc \rarr \Bc$ is injective.    
\end{proof}

Using Lemma 4 allows us to present the following mixed integer linear program which solves for $\rmax(\D)$:

\begin{theorem}
\label{thm:rrobust}
    Let $\D= (\V,\E)$ be an arbitrary nonempty, nontrivial, simple digraph and let $L$ be the Laplacian matrix of $\D$. The maximum $r$-robustness of $\D$, denoted $\rmax(\D)$, is obtained by solving the following minimization problem:
    \begin{align}
            \rmax(\D) =& \hspace{.5em} \underset{\bm b^1,\bm b^2}{\min}
            & & \max \pth{ \max_i \pth{L_i \bm b^1}, \max_j \pth{L_j \bm b^2} } \nonumber \\
            & \hspace{-.5em} \text{\textup{subject to}}
            & & \bm b^1 + \bm b^2 \gleq \bm 1 \nonumber \\
            & & & 1 \leq \bm 1^T \bm b^1 \leq (n-1) \nonumber \\
            & & & 1 \leq \bm 1^T \bm b^2 \leq (n-1) \label{eq:finalprobS0} \nonumber \\
            & & & \bm b^1, \bm b^2 \in \{0,1\}^n.
    \end{align}
\end{theorem}


\begin{proof}
From Lemmas \ref{lem:rrobalt} and \ref{lem:maxequal}  we have
\begin{align}
            \rmax(\D) = \nonumber \\
             \underset{S_1,S_2 \in \Pc(\V)}{\min}
            & & \max \pth{ \max_i \pth{L_i \bm \sigma(S_1)}, \max_j \pth{L_j \bm \sigma(S_2)} } \nonumber \\
             \text{subject to}
            & & |S_1| > 0,\ |S_2| > 0,\ |S_1 \cap S_2| = 0, \nonumber
    \end{align}
for $i,j \in \{1,\ldots,n\}$. As per Remark \ref{rmk:implicit}, the definition of $\mathcal{T}$ can be used to make the constraints implicit:
\begin{align}
        \label{eq:rrobT}
            &\rmax(\D) = \nonumber \\
            &\underset{(S_1,S_2) \in \mathcal{T}}{\min}
             \max \pth{ \max_i \pth{L_i \bm \sigma(S_1)}, \max_j \pth{L_j \bm \sigma(S_2)} },           
\end{align}
for $i,j \in \{1,\ldots,n\}$. Since $\Sigma: \mathcal{T} \rarr \Bc$ is a bijection by Lemma \ref{lem:bijec}, \eqref{eq:rrobT} is equivalent to
\begin{align}
    \label{eq:implicitB}
            \rmax(\D) =& \underset{(\bm b^1, \bm b^2) \in \Bc}{\min}
            & & \max \pth{ \max_i \pth{L_i \bm b^1}, \max_j \pth{L_j \bm b^2} } \nonumber \\
            & & & i,j \in \{1,\ldots,n\}.
\end{align}
Making the constraints of \eqref{eq:implicitB} explicit yields \eqref{eq:finalprobS0}.    
\end{proof}

This minimization problem can actually be reformulated to an mixed integer linear programming problem where the objective and all constraint functions are affine except for the integer constraint. This is shown in the following corollary:

\begin{cor}
\label{cor:equivgen}
Let $\D = (\V,\E)$ be an arbitrary nonempty, nontrivial, simple digraph, and let $L$ be the Laplacian matrix of $\D$. The maximum $r$-robustness of $\D$, denoted $\rmax(\D)$, is obtained by solving the following mixed integer linear program:
\begin{align}
            \rmax(\D)=
            \hspace{.5em} \underset{t, \bm b}{\min \hspace{.5em}}
            &\ t \nonumber\\
             \hspace{.5em} \text{\textup{subject to \hspace{.5em}}}
             & t \in \R,\ \bm b \in \Z^{2n} \nonumber \\
             &0 \leq t \nonumber\\
             &\bmx{L & \bm 0 \\ \bm 0 & L} \bm b \gleq t \bmx{\bm 1 \\ \bm 1} \nonumber\\
              &\bm 0 \gleq \bm b \gleq \bm 1 \nonumber\\
              &\bmx{I_{n\times x} & I_{n \times n}}\bm b \gleq \bm 1 \nonumber\\
              &1 \leq  \bmx{\bm 1^T & \bm 0} \bm b \leq n-1 \nonumber\\
              &1 \leq \bmx{\bm 0 & \bm 1^T} \bm b \leq n-1 \label{eq:probaffine}
\end{align}
\end{cor}

\begin{proof}
    The variables $\bm b^1$ and $\bm b^2$ from \eqref{eq:finalprobS0} are combined into the variable $\bm b \in \Z^{2n}$ in \eqref{eq:probaffine}; i.e. $\bm b =\bmxs{(\bm b^1)^T & (\bm b^2)^T}^T $. The first and fourth constraints of \eqref{eq:probaffine} restrict $\bm b \in \{0,1 \}^{2n}$. Next, it can be demonstrated \cite[Chapter 4]{boyd2004convex} that the formulation ${\min_{\bm x}} \max_i (x_i)$ is equivalent to
    \begin{equation*}
        \begin{aligned}
        & \underset{t, \bm x}{\min}
        & & t \\
        & \text{subject to}
        & &0 \leq t,\ \bm x \gleq t \bm 1. \\
        \end{aligned}
    \end{equation*}
    Reformulating the objective of the RHS of \eqref{eq:finalprobS0} in this way yields the objective, and second and third constraints of \eqref{eq:probaffine}:
    \begin{equation}
        \begin{aligned}
            & \underset{t, \bm b}{\min}
            & & t \\
            & \text{subject to}
            &  &0 \leq t \\
            & & &\bmx{L & \bm 0 \\ \bm 0 & L} \bm b \gleq t \bmx{\bm 1 \\ \bm 1}. \\
        \end{aligned}
    \end{equation}
    The fifth, sixth, and seventh constraints of \eqref{eq:probaffine} restrict $(\bm b^1, \bm b^2) \in \Bc$ and are simply a reformulation of the first three constraints in \eqref{eq:finalprobS0}.
\end{proof}

%
%




\subsection{Discussion}



Since $r$-robustness is equivalent to $(r,1)$-robustness,\footnote{See section VII-B of \cite{leblanc2013resilient}.} the solution to the optimization problem in Theorem \ref{thm:rrobust} determines the maximum $r$ for which the graph is $(r,1)$-robust. As per \cite{leblanc2013resilient,leblanc2013algorithms}, in general the parameter $r$ has higher precedence than $s$ when ordering a set of graphs by robustness. In addition, $(r'+s'-1)$-robustness implies $(r',s')$-robustness \cite{leblanc2012thesis},\footnote{E.g. $(2F+1)$-robustness for $F \in \Z_+$ implies $(F+1,F+1)$-robustness.} and so a certain degree of $(r,s)$-robustness can be inferred from knowing the maximum value of $r$.

When comparing the optimization problem in Theorem \ref{thm:rrobust} with the algorithms in \cite{leblanc2013algorithms}, it is important to note that those in \cite{leblanc2013algorithms} determine both of the parameters $r$ and $s$ for which graphs are $(r,s)$-robust. Since the method in Theorem \ref{thm:rrobust} only determines the largest $r$ for which a digraph is $(r,1)$-robust (with $s$ fixed at 1), it cannot be directly compared to Algorithm 3.2, $DetermineRobustness(\A(\D))$ in \cite{leblanc2013algorithms}. 
However, by replacing each initialization $s \larr n$ in $DetermineRobustness$ with the initialization $s \larr 1$, a modified algorithm is obtained which only determines $(r,1)$-robustness and can be directly compared with the method in Theorem \ref{thm:rrobust}. This modified algorithm is presented as Algorithm \ref{alg:modified} and serves as the benchmark against which we compare the performance of the MILP formulation presented in this paper (see Section \ref{sec:simulations}).
In addition, the initialization condition $r \larr \min(\delta^{\text{in}}(\D),\ceil{n/2})$ in $DetermineRobustness$ incorrectly classifies some rooted outbranchings as 0-robust (those with in-degree of the root being 0), when they are actually 1-robust (see Lemma 7 of \cite{leblanc2013resilient}). We have revised the initialization of $r$ accordingly.
The reader is referred to \cite{leblanc2013algorithms} for the definition of the function R\textsc{obustholds}$(A,S_1,S_2,r,s)$. In short, the function returns the boolean \textbf{true} if the number of $r$-reachable nodes from $S_1$ and $S_2$ is at least $s$, and \textbf{false} otherwise.

\begin{algorithm}
\caption{\small{Modified version of D\textsc{etermineRobustness}}}\label{alg:modified}
\begin{algorithmic}[1]
\State $r \leftarrow \min\pth{\max\pth{\delta^{\text{in}}(\D),1},\ceil{\frac{n}{2}}}$
\State $s \leftarrow 1$ \Comment{\emph{(Different than Alg. 3.2 in \cite{leblanc2013algorithms})}}
\State \textbf{comment:} $\delta^{\text{in}}(\D)$ is the min. in-degree of nodes in $\D$
\ForAll{$k \larr 2$ to $n$}
    \State \hspace{-1.9em} \textbf{comment:} $\mathcal{K}_k$ is the set of $C(n,k)$ unique subsets of $\V$
    \ForAll {$K_i \in \mathcal{K}_k\ (i = 1,2,\ldots,C(n,k))$}
    \ForAll {$P_j \in \mathcal{P}_{K_i}\ (j = 1,2,\ldots,2^{k-1}-1)$}
    \State \hspace{-2.2em}  \begin{tabular}{l l} \textbf{comment:} &$\mathcal{P}_{K_i}$ is set of partitions of $K_i$ into $S_1$ \\
    &and $S_2$ \end{tabular}
    \State ${isRSRobust \larr \text{R\textsc{obustholds}}}(A,S_1,S_2,r,s)$
    \If{($isRSRobust == $ \textbf{false}) \textbf{and} $s > 0$}
        \State $s \larr s -1$
    \EndIf
    \While{$isRRobust == $ \textbf{false} \textbf{and} $(r > 0)$}
        \While{$isRSRobust == $ \textbf{false} \textbf{and} $(s > 0)$}
            \State $isRRobust$
            \State \hspace{2em}$\larr \text{R\textsc{obustholds}}(A,S_1,S_2,r,s)$
            \If{\textbf{not} $isRSRobust$}
                \State $s \larr s-1$
            \EndIf
        \EndWhile
        
        \If{$isRSRobust == $ \textbf{false}}
            \State $r \larr r-1$
            \State $s \larr 1$  \Comment{\emph{(Different than Alg. 3.2 in \cite{leblanc2013algorithms})}}
        \EndIf
    \EndWhile
    
    \If{$r == 0$}
        \State \textbf{return} $r$
    \EndIf 
    \EndFor
    \State \textbf{end for}
    \EndFor
    \State \textbf{end for}
\EndFor
\State \textbf{end for}
\end{algorithmic}
\end{algorithm}

There are only two ways Algorithm \ref{alg:modified} will terminate: either the algorithm finds an $S_1 $ and $S_2 $ pair which are both 0-reachable, or the algorithm checks all possible unique pairs of subsets. Any subsets $S_1,S_2$ found such that $(S_1,S_2) \in \mathcal{T}$ and $\max(\Rc(S_1),\Rc(S_2)) < \beta$ for $\beta \in \Z_+$ is a certificate that the graph is \emph{not} $\beta$-robust \cite{zhang2015notion}; hence $r_{\max}(\D) < \beta$. It is only possible for Algorithm \ref{alg:modified} to tighten the \emph{upper bound} on $r_{\max}(\D)$ unless all pairs of relevant subsets are checked.



On the other hand, since the optimal value of Corollary \ref{cor:equivgen} is equal to $r_{\max}(\D)$ it is possible for a \emph{lower} bound on the robustness of $\D$ to be tightened over time by using a branch-and-bound (B\&B) algorithm \cite{wolsey2007mixed,vanderbei2015linear}. A lower bound on $r_{\max}(\D)$ is often more useful than an upper one, since $r_{\max}(\D) \geq \gamma$ implies that $\D$ is $r$-robust for all $0 \leq r \leq \gamma$.
The crucial advantage of B\&B algorithms is that both a global upper bound \emph{and lower bound} on the objective value are calculated and iteratively tightened as successive convex relaxations of the optimization problem are solved. When the gap between these bounds becomes zero, optimality is obtained. However, the search can also be terminated if the lower bound on the objective reaches a sufficiently high value. In context of robustness determination,
this therefore introduces the possibility of calculating approximate lower bounds on $r_{\max}(\D)$ for arbitrary nonempty, nontrivial, simple digraphs (and undirected graphs) without needing to fully solve for $r_{\max}(\D)$. A more detailed examination of the convergence rate for this lower bound is left for future work.

\section{Simulations}
\label{sec:simulations}

Simulations were conducted to compare the computation time of our formulation against Algorithm \ref{alg:modified}, the modified version of $DetermineRobustness$ proposed in \cite{leblanc2013algorithms}. Computations were performed in MATLAB 2018a on a desktop computer with an Intel Core i7-6700 CPU (3.40GHz) capable of handling 8 threads, and with 31GB of RAM.
The simulations tested the time required for the algorithms to calculate the \emph{exact} maximum $r$-robustness for various types of graphs. The algorithms tested were Algorithm \ref{alg:modified}, which is a modification of $DetermineRobustness$ in \cite{leblanc2013algorithms}, and the proposed formulation in Corollary \ref{cor:equivgen} solved using MATLAB's \emph{intlinprog} function. Four classes of random graphs were tested: Erd\H{o}s-R\'enyi random undirected graphs, directed random graphs, $k$-in directed random graphs, and $k$-out directed random graphs \cite{bollobas2001random}.\footnote{$k$-in random graphs are constructed in the same manner as $k$-out random graphs, but with the edge directions reversed.} 
Various values of $n$ were selected ranging from $7$ to $15$, and for each value of $n$ the algorithms were tested on 100 graphs. Additionally, the proposed MILP formulation was tested on random digraphs with values of $n$ ranging from $18$ to $30$, where 100 graphs were tested for each value of $n$.  For Erd\H{o}s-R\'enyi graphs and random digraphs, simulations were performed for edge formation probabilities $p \in \{.3, .5, .8 \}$. For the $k$-in and $k$-out random digraphs, simulations were performed for $k \in \{3,4,5\}$. Due to space constraints, only graphs for values $p = .5$, $p=.8$, and $k = 4$ are shown.
In all trials for the $n$ values where both algorithms were tested, the maximum $r$-robustness of the MILP formulation found for the graph was \emph{exactly} equal to the maximum $r$-robustness returned by the exhaustive search method Algorithm 1. 
The computation time of both algorithms demonstrate an exponential trend. However, for $n \geq 8$ the average computation time for the MILP algorithm is clearly less than the average computation time for Algorithm \ref{alg:modified}. 



\begin{figure}
    \centering
    \includegraphics[width=\figsc\columnwidth]{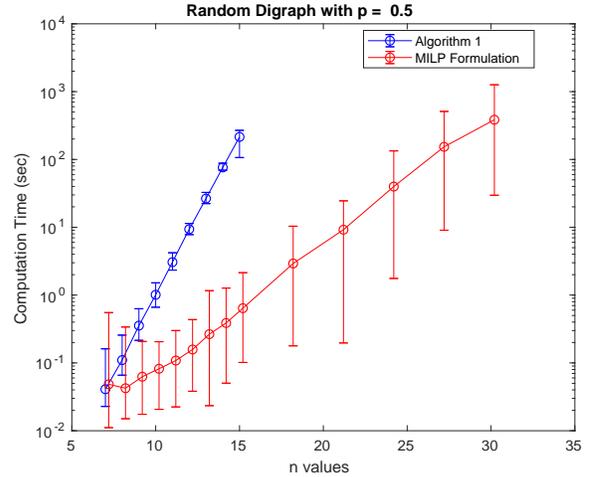}
    \caption{Computation time for determining maximum $r$-robustness of random directed graphs with parameter $p = 0.5$. Note the logarithmic scale on the y-axis. The vertical lines indicate the spread between max and min times, respectively, for each value of $n$.}
\end{figure}

\begin{figure}
    \centering
    \includegraphics[width=\figsc\columnwidth]{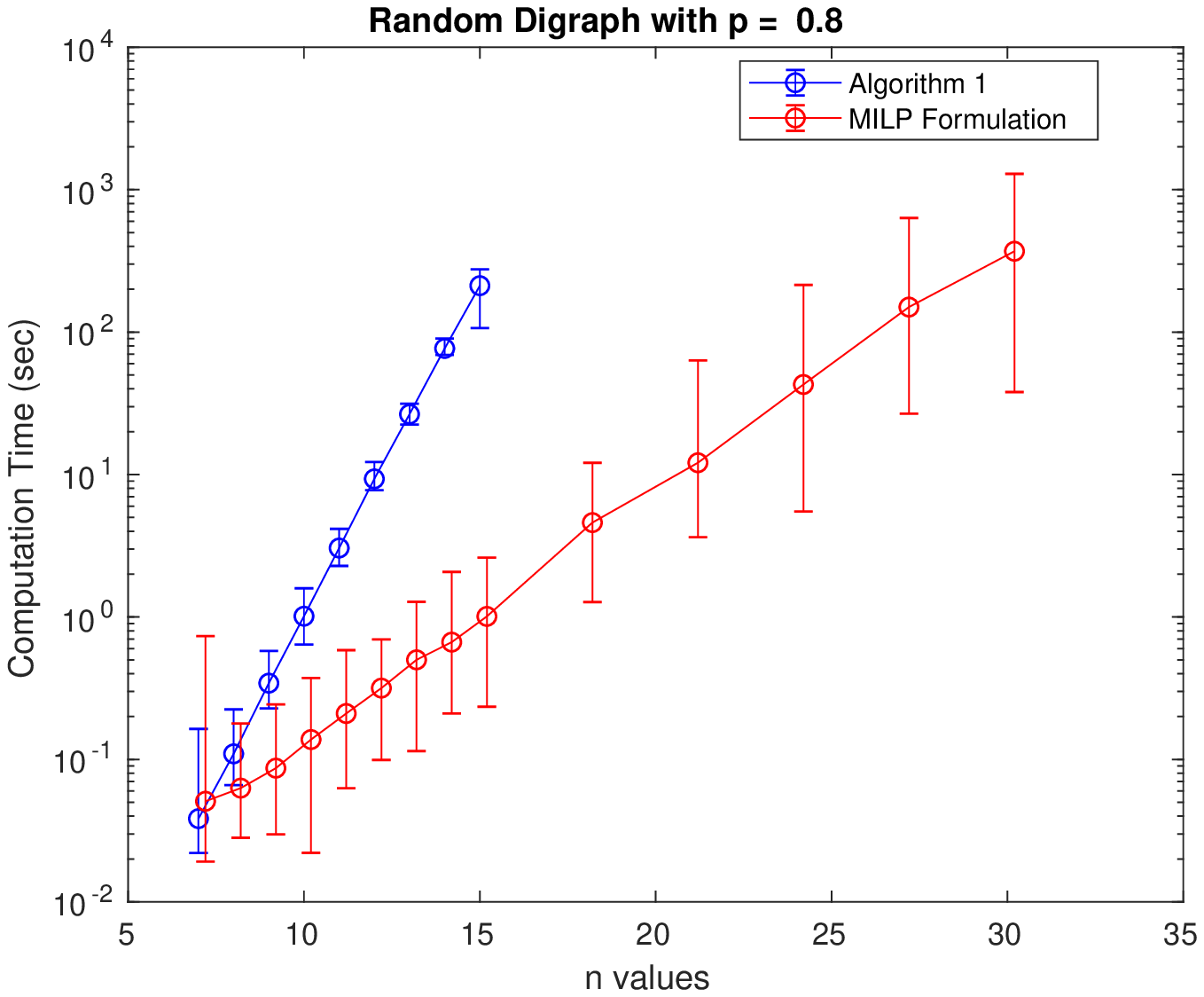}
    \caption{Computation time for determining maximum $r$-robustness of random directed graphs with parameter $p = 0.8$.}
\end{figure}

\begin{figure}
    \centering
    \includegraphics[width=\figsc\columnwidth]{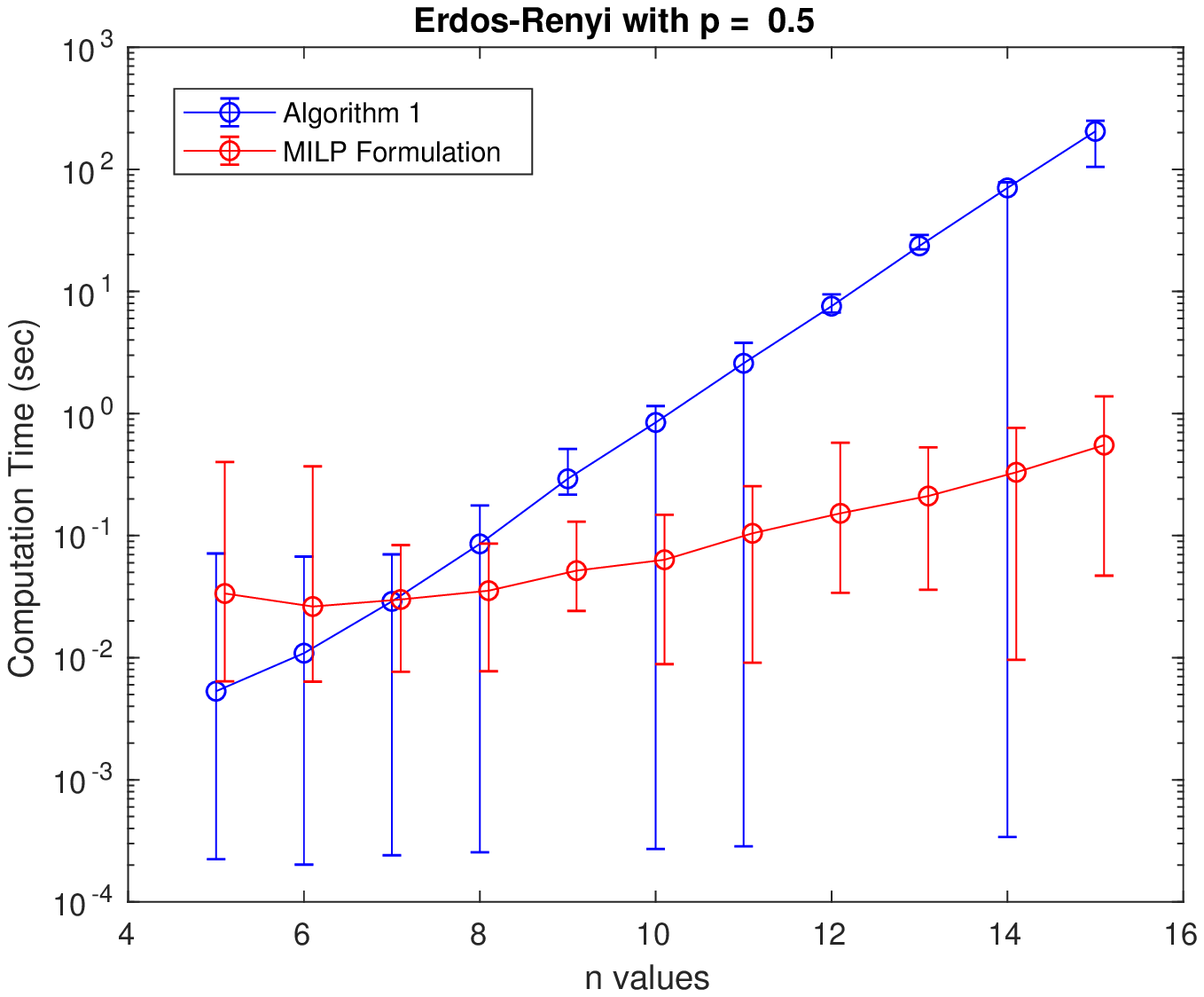}
    \caption{Computation time for determining maximum $r$-robustness of \Erdosrenyi random graphs with parameter $p = 0.5$.}
    \label{fig:Erdosp05}
\end{figure}

\begin{figure}
    \centering
    \includegraphics[width=\figsc\columnwidth]{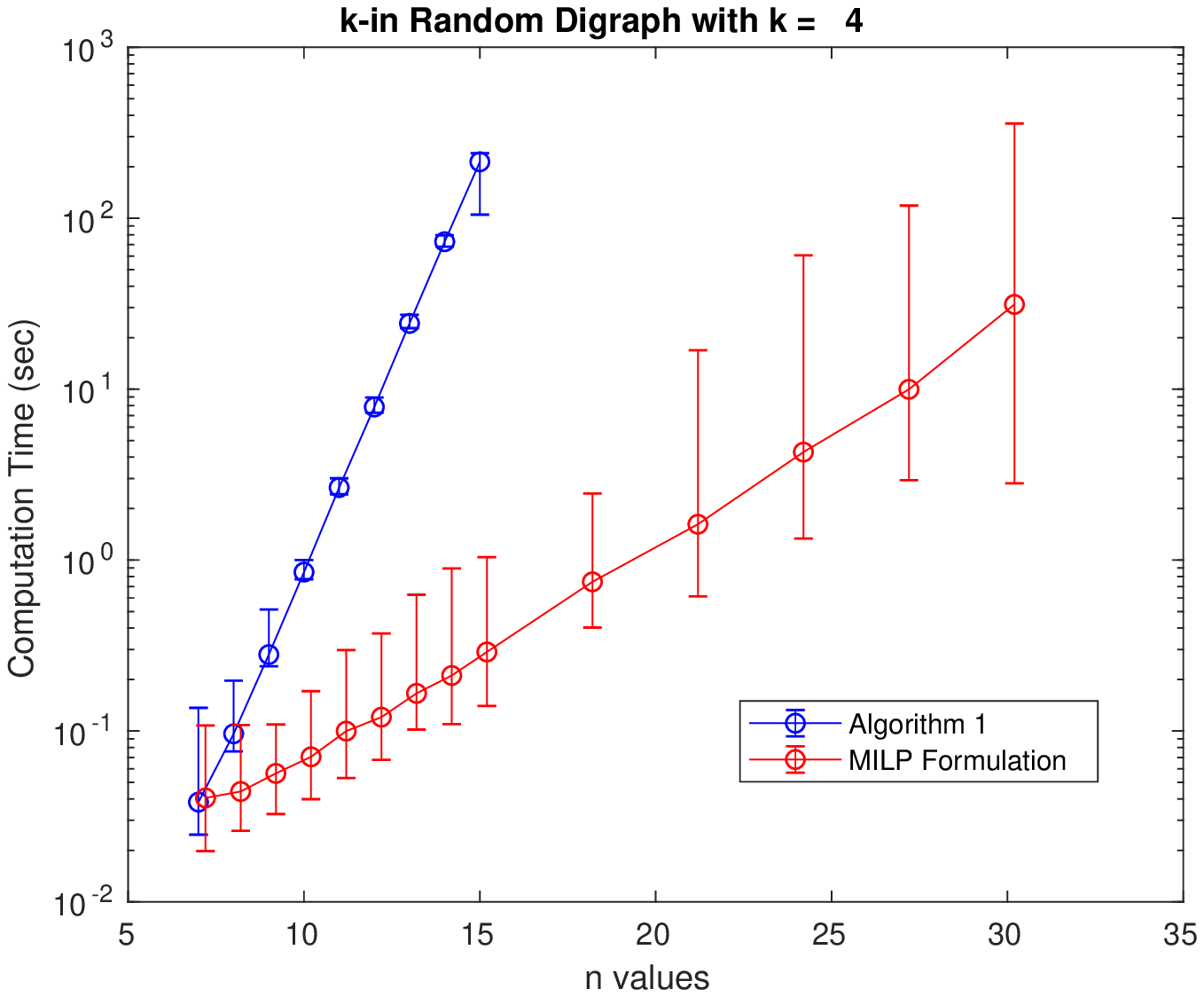}
    \caption{Computation time for determining maximum $r$-robustness of $k$-in random directed graphs with parameter $k = 0.5$. }
    \label{fig:kinrand_k04}
\end{figure}

\begin{figure}
    \centering
    \includegraphics[width=\figsc\columnwidth]{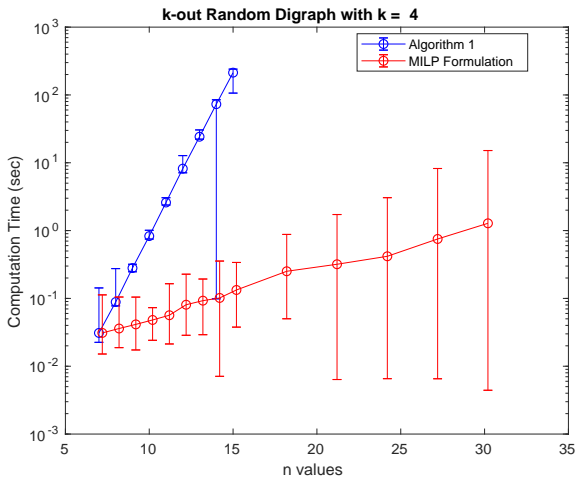}
    \caption{Computation time for determining maximum $r$-robustness of $k$-out random directed graphs with parameter $k = 0.5$. }
    \label{fig:koutrand_k04}
\end{figure}


\section{Conclusion}
\label{sec:conclusion}
This paper formulated the problem of determining the largest integer $r$ for which an arbitrary simple digraph is $r$-robust as a mixed integer linear program. This formulation was shown to have affine objective and constraints, except the integer constraint. Simulations suggest that this formulation demonstrates reduced computation time as compared to prior algorithms. Future work will include investigating techniques from the integer programming literature to further improve the efficiency of this method.

\bibliographystyle{IEEEtran}

\bibliography{ACC2019.bib}

\end{document}

%% file: JUcmds.tex
\newtheorem{define}{Definition}
\newtheorem{theorem}{Theorem}
\newtheorem{lemma}{Lemma}
\newtheorem{cor}{Corollary}
\newtheorem{remark}{Remark}

\newcommand{\Bc}{\mathcal{B}}
\newcommand{\Rc}{\mathcal R}
\newcommand{\R}{\mathbb R}

\newcommand{\N}{\mathcal{N}}
\newcommand{\D}{\mathcal{D}}
\newcommand{\V}{\mathcal{V}}

\newcommand{\G}{\mathcal{G}}
\newcommand{\E}{\mathcal{E}}

\newcommand{\A}{\mathcal{A}}

\newcommand{\Z}{\mathbb{Z}}

\newcommand{\Pc}{\mathcal{P}}
\newcommand{\Tc}{\mathcal{T}}

\newcommand{\bmx}[1]{\begin{bmatrix}#1\end{bmatrix}} 
\newcommand{\pth}[1]{\left(#1\right)} 
\newcommand{\nrm}[1]{\left \lVert#1\right \rVert} 

\newcommand{\bmxs}[1]{\begin{bsmallmatrix}#1\end{bsmallmatrix}}
\newcommand{\pmxs}[1]{\begin{psmallmatrix}#1\end{psmallmatrix}}

\newcommand{\gleq}{\preceq} 

\DeclarePairedDelimiter{\ceil}{\lceil}{\rceil}
\DeclarePairedDelimiter{\floor}{\lfloor}{\rfloor}
\DeclarePairedDelimiter{\abs}{\lvert}{\rvert}

\newcommand{\rarr}{\rightarrow} 
\newcommand{\larr}{\leftarrow} 

\newcommand{\Erdosrenyi}{Erd\H{o}s-R\'enyi } 

\makeatletter
\let\oldceil\ceil
\def\ceil{\@ifstar{\oldceil}{\oldceil*}}
\makeatother
\makeatletter
\let\oldfloor\floor
\def\floor{\@ifstar{\oldfloor}{\oldfloor*}}
\makeatother
\makeatletter
\let\oldnorm\norm
\def\norm{\@ifstar{\oldnorm}{\oldnorm*}}
\makeatother
\makeatletter
\let\oldabs\abs
\def\abs{\@ifstar{\oldabs}{\oldabs*}}
\makeatother